\documentclass[conference, 10pt]{IEEEtran}

\ifCLASSINFOpdf

\else

\fi

\usepackage{array}

\usepackage{amsthm}
\usepackage{mdwmath}
\usepackage{mdwtab}
\usepackage[utf8]{inputenc}
\usepackage{multicol}
\usepackage{eqparbox}
\usepackage{wrapfig,lipsum,booktabs}
\usepackage{graphicx}
\usepackage{epstopdf} 

\usepackage{color}
\usepackage{placeins}
\usepackage{float}
\usepackage{tabularx,colortbl}
\usepackage{amsmath}
\usepackage{amssymb}

\usepackage[square,sort,comma,numbers]{natbib}

\begin{document}
%
\title{Resource allocation in OFDMA networks with half-duplex and imperfect full-duplex users
}

\author{\IEEEauthorblockN{Peyman Tehrani\IEEEauthorrefmark{1},
Farshad Lahouti\IEEEauthorrefmark{1}\IEEEauthorrefmark{2},
Michele Zorzi\IEEEauthorrefmark{3}}

\IEEEauthorblockA{\IEEEauthorrefmark{1}School of Electrical and Computer Engineering,
College of Engineering, University of Tehran, Iran}
\IEEEauthorblockA{\IEEEauthorrefmark{2}Electrical Engineering Department, California Institute of Technology, USA}
\IEEEauthorblockA{\IEEEauthorrefmark{3}Department of Information Engineering, University of Padova, Italy \\
Emails: peymantehrani@ut.ac.ir, lahouti@caltech.edu.com or lahouti@ut.ac.ir, zorzi@dei.unipd.it}
}


\maketitle

\begin{abstract}
Recent studies indicate the feasibility of in-band full-duplex (FD) wireless communications, where a wireless radio transmits and receives simultaneously in the same band. Due to its potential to increase the capacity, analyzing the performance of a cellular network that contains full-duplex devices is crucial. In this paper, we consider maximizing the weighted sum-rate of downlink and uplink of a single cell OFDMA network which consists of an imperfect FD base-station (BS) and a mixture of half-duplex and imperfect full-duplex mobile users. To this end, the joint problem of sub-channel assignment and power allocation is investigated and a two-step solution is proposed. A heuristic algorithm to allocate each sub-channel to a pair of downlink and uplink users with polynomial complexity is presented. The power allocation problem is convexified based on the difference of two concave functions approach, for which an iterative solution is obtained. Simulation results demonstrate that when all the users and the BS are perfect FD nodes the network throughput could be doubled, Otherwise, the performance improvement is limited by the inter-node interference and the self-interference. We also investigate the effect of the self-interference cancellation capability and the percentage of FD users on the network performance in both indoor and outdoor scenarios.

\end{abstract}

\IEEEpeerreviewmaketitle

\section{Introduction}

In wireless communications, separation of transmission and reception in time or frequency has been the standard practice so far. However, through simultaneous transmission and reception in the same frequency band, wireless full-duplex has the potential to double the spectral efficiency. Due to this substantial gain, full-duplex technology has recently attracted noticeable interest in both academic and industrial worlds. The main challenge in full-duplex (FD) bidirectional communication is self-interference (SI) cancellation. In recent years, many attempts have been made to cancel the self-interference signal. In \citep{bharadia2013full}, it is shown that $110$ dB SI cancellation is achievable, and by jointly exploiting analog and digital techniques, SI may be reduced to the noise floor.

A full-duplex physical layer in cellular communications calls for a re-design of higher layers of the protocol stack, including scheduling and resource allocation algorithms. In \cite{goyal2013analyzing}, the performance of an FD-based cellular system is investigated and an analytic model to derive the average uplink and downlink channel rate is provided. A resource allocation problem for an FD heterogeneous OFDMA network is considered in \cite{sultanmode}, in which the macro BS and small cell access points operate in either  FD or HD MIMO mode, and all mobile nodes operate in HD single antenna mode. In \cite{di2014radio}, using matching theory, a sub-channel allocation algorithm for an FD OFDMA network is proposed. In both \cite{sultanmode} and \cite{di2014radio} only a single sub-channel is assigned to each of the uplink users in which they transmit with constant power. Resource allocation solutions are proposed in \cite{nam2015joint} and \cite{namradio} for FD OFDMA networks with perfect FD nodes (SI is canceled perfectly).
 
In this paper, we consider a general resource allocation problem in an OFDMA-based network consisting of an imperfect FD BS and both HD and imperfect FD users. We aim to maximize the weighted sum-rate of this network in the uplink and the downlink by joint sub-channel assignment and power allocation. To be more realistic, imperfect SI cancellation in FD devices is assumed and FD nodes suffer from their SI. Since our model is general in the sense that we assign both uplink and downlink weights to the users, HD users are allowed to transmit and receive in different sub-channels. Hence, when a node is only an uplink (downlink) user, its downlink (uplink) weight is set to zero. A contribution of our work is to consider the presence of a mixture of FD and HD users, which enables us to quantify the percentage of FD users needed to capture the full potential of FD technology in wireless OFDMA networks. We also investigate the effect of the SI cancellation level on the network performance, which has never been considered in related works.

The remainder of this paper is organized as follows. In Section \ref{system_model}, the system model is given and the optimization problem is formulated. In Section \ref{subchannel allocation}, a sub-channel allocation algorithm for selecting the best pair in each sub-channel is presented. Power allocation is considered in Section \ref{power_allocation}. Numerical results for the proposed resource allocation methods are shown in Section \ref{simulation_results}. Finally, the paper is concluded in Section \ref{conclusion}.

\section{System Model and Problem Statement} \label{system_model}
We consider a single cell network that consists of a full-duplex base-station (BS) and a total of $K$ half-duplex and full-duplex users. For communications between the nodes and the BS, we assume that an OFDMA system with $N$ sub-channels is used. All sub-carriers are assumed to be perfectly synchronized, and so there is no interference between different sub-channels. Since the base-station operates in full-duplex mode, it can transmit and receive simultaneously in each sub-channel. In each timeslot the base-station is to properly allocate the sub-channels to the downlink or uplink of appropriate users and also determine the associated transmission power in an optimized manner. We assume that the base-station and the FD users are imperfect full-duplex nodes that suffer from self-interference. We define a self-interference cancellation coefficient to take this into account in our model and denote it by $0\leq\beta\leq1$, where $\beta=0$  indicates that SI is canceled perfectly and $\beta=1$  means no SI cancellation. For simplicity, we assume the same self-interference cancellation coefficient for BS and FD users, but consideration of different coefficients would be possible. In this paper, the goal is to maximize the weighted sum-rate of downlink and uplink users with a total power constraint at the base-station and a transmission power constraint for each user.

\begin{figure}
  \includegraphics[width=\linewidth, height=4cm]{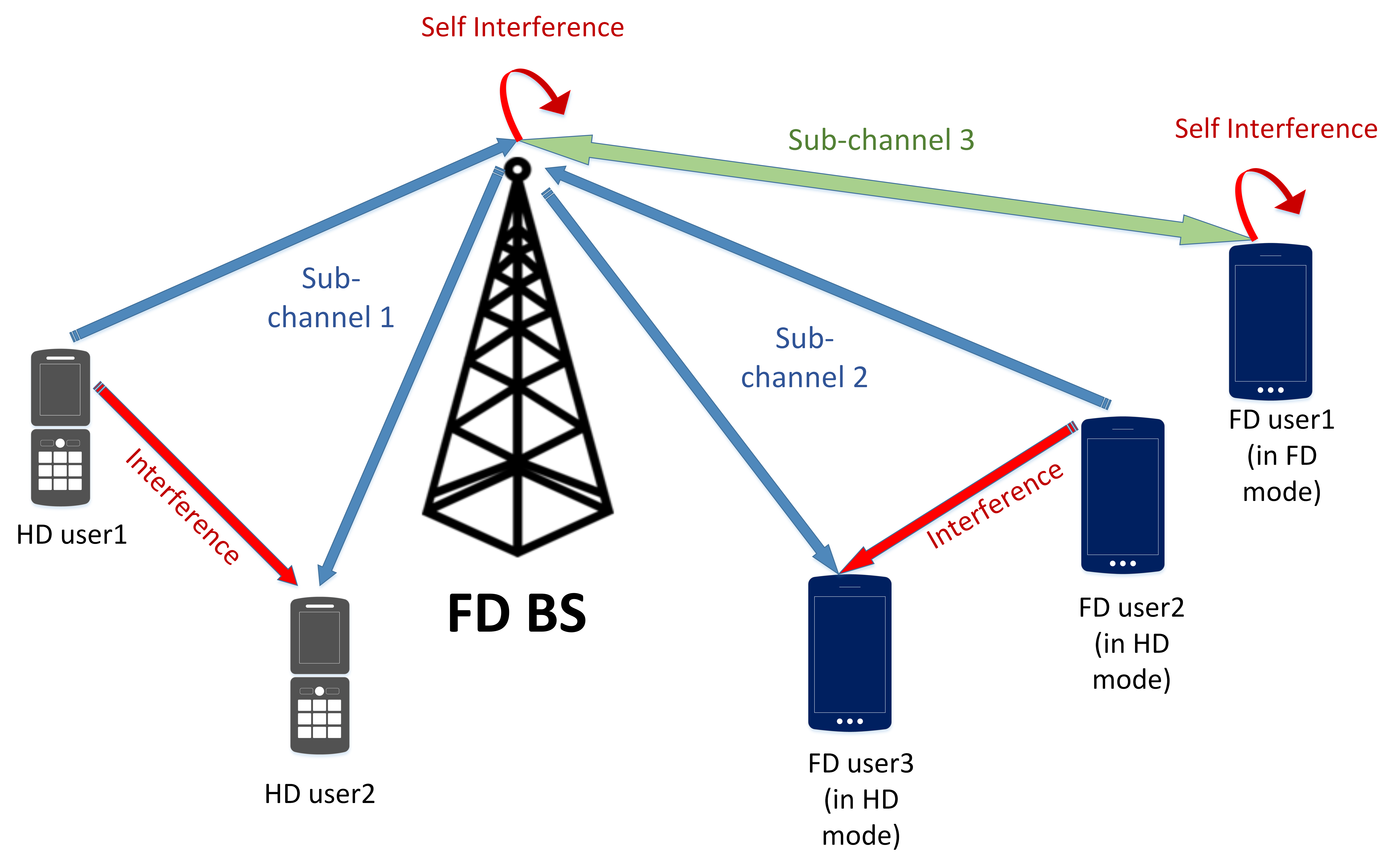}
  \caption{A single cell OFDMA full-duplex network that contains an imperfect full-duplex base-station and multiple half-duplex and full-duplex mobile nodes. Due to the full-duplex nature of this network, the base-station suffers from its self interference, and the uplink nodes cause interference to their co-channel downlink nodes}
\end{figure}

 We define the downlink weighted sum-rate as
 \vspace{-3mm}
\begin{equation}\label{DLRateQ}
R_d=\sum_{k=1}^{K}\sum_{n\in S_{k,d}}w_k\log \bigg(1+ \frac{g_k(n)p_{k,d}(n)}{N_k+I_{k,j}(n)p_{j,u}(n)} \bigg)
\end{equation}
And the uplink weighted sum-rate as 
\begin{equation}
R_u=\sum_{j=1}^{K}\sum_{n\in S_{j,u}}v_j\log \bigg(1+ \frac{g_j(n)p_{j,u}(n)}{N_0+\beta p_{k,d}(n)}\bigg)
\end{equation}
The variables used in the above equations are introduced in Table I. We assume here that the channel is reciprocal, i.e., uplink and downlink channel gains are the same. We further assume that the receiver noise powers in different sub-channels are the same. The term $I_{k,j}(n)p_{j,u}(n)$ in \eqref{DLRateQ} denotes the interference: When user $k$ is a FD device and both downlink and uplink of sub-channel $n$ are allocated to it $(j=k)$, $I_{k,j}(n)=\beta$, else $I_{k,j}(n)=g_{k,j}(n)$ is the channel gain between uplink user $j$ and downlink user $k$. We assume that the base-station knows all the channel gains, noise powers, the SI cancellation coefficient and weights assigned to the downlink and uplink of all users. 

Let $P_0$ and $P_k$ denote the maximum available transmit power for the base-station and for user $k$, respectively. Then the proposed design optimization problem, denoted by \textbf{\textit{P}1}, can be formulated as follows

\begin{table}[t!]
 \caption{Main Parameters and Variables}
 \begin{tabular}{|c | m{24em}|}

 \hline
 $w_k$ & weight assigned to the downlink  of user $k$ \\
 \hline
 $v_k$ & weight assigned to the uplink of user $k$ \\
 \hline
  $p_{k,d}(n)$ &  transmission power from BS to user $k$ on sub-channel $n$ \\
 \hline
  $p_{j,u}(n)$ & transmission power from user $j$ to BS on sub-channel $n$ \\
 \hline
  $ N_k$  & Gaussian noise variance at the receiver of user $k$ \\
 \hline
  $N_0$ & Gaussian noise variance at the base-station receiver\\
 \hline
  $S_{k,d}$ & set of sub-channels allocated to user $k$ for downlink transmission\\
 \hline
  $S_{j,u}$ & set of sub-channels allocated to user $j$ for uplink transmission \\
 \hline
  $\beta$ & self-interference cancellation coefficient \\
 \hline
 $g_{k}(n)$ & channel gain between BS and user $k$ on sub-channel $n$ \\
 \hline
 $g_{k,j}(n)$ & channel gain between user $j$ and $k$ on sub-channel $n$ \\
 \hline
 $I_{k,j}(n)$ & equal to $\beta$ when $j=k$, and to $g_{k,j}(n)$ otherwise\\
 \hline
 $P_0$ & maximum available transmit power at BS\\
 \hline
 $P_k$ & maximum available transmit power at user $k$ \\
 \hline
 \end{tabular}
\end{table}

\begin{align}
\textbf{\textit{P}1}: \operatorname*{maximize}_{p_{k,d},p_{j,u},S_{j,u},S_{k,d}} \qquad R_d+R_u\\
\text{subject to }   \sum_{k=1}^{K}\sum_{n\in S_{k,d}}p_{k,d}(n) \leq P_0\\
\sum_{n\in S_{j,u}}p_{j,u}(n) \leq P_j \quad \forall j\\
p_{j,u}(n),p_{k,d}(n)\geq 0 \quad \forall j,k,n\\
 S_{i,d}\cap S_{j,d}=\phi \quad \forall i\neq j\\
 S_{i,u}\cap S_{j,u}=\phi \quad \forall i\neq j\\
\cup_{j=1}^{K}  \ S_{j,u} \subseteq \{1,2,...,N\} \\
\cup_{k=1}^{K} \ S_{k,d} \subseteq \{1,2,...,N\} \\
S_{k,u} \cap S_{k,d}=\phi \quad \text{if user} \ k \  \text{is} \ \text{HD}
\end{align}
where (4) and (5) indicate the power constraint on the BS and the users, respectively. Constraint (6) shows the non-negativity feature of powers; (7) and (8) come from the fact that a sub-channel cannot be allocated to two distinct users simultaneously; (9) and (10) indicate that we have no more than $N$ sub-channels, and the last constraint accounts for the half-duplex nature of the HD users.

The general resource allocation problem presented is combinatorial in nature because of the channel allocation issue and addressing it together with power allocation in an optimal manner is challenging, especially as the number of users and sub-channels grow. Moreover, the non-convexity of the rate function makes the power allocation problem itself challenging even for a fixed sub-channel assignment. Here, we invoke a two step approximate solution. First, we  determine the allocation of downlink and uplink sub-channels to users and then determine the transmit power of the users and the base-station on their allocated sub-channels. In other words, we first specify the sets $S_{k,d}$ and $S_{j,u}$ and then determine the variables $p_{j,u}(n)$, $p_{k,d}(n)$. In the next Section, we introduce our sub-channel allocation algorithm. 

\section{Sub-channel Allocation} \label{subchannel allocation}
 
The sub-channel allocation problem, denoted by \textbf{\textit{P}2}, can be formulated as follows 
\begin{align*}
\textbf{\textit{P}2}: & \operatorname*{maximize}_{S_{j,u},S_{k,d}} \qquad &R_d+R_u\\
& \text{subject to} \qquad &\text{(7)-(11)}
\end{align*}
To solve the problem \textbf{\textit{P}2}, we should first solve the following power allocation problem, denoted by \textbf{\textit{P}3}, to maximize the weighted sum-rate in a single sub-channel and for a fixed pair of uplink and downlink users. Since a single sub-channel is being considered in \textbf{\textit{P}3}, we have dropped the variable $n$ in the notation.
\begin{align*}
\textbf{\textit{P}3}: \operatorname*{max}_{p_{k,d},p_{j,u}} L(p_{k,d},p_{j,u})=&w_k\log(1+\frac{g_kp_{k,d}}{N_k+I_{k,j}p_{j,u}}) 
\\
+&v_j\log(1+ \frac{g_jp_{j,u}}{N_0+\beta p_{k,d}})
\end{align*}
\begin{align}
&0\leq p_{k,d}\leq P_{max1} \\
&0\leq p_{j,u}\leq P_{max2}
\end{align}
Here, $P_{max1}$ and $P_{max2}$ are the maximum allowable transmit powers.
\newtheorem{prop}{Proposition}
\begin{prop}
For a fixed downlink user $k$ and uplink user $j$, the optimal pair of powers $(p_{k,d}^*,p_{j,u}^*)$ that optimizes \textbf{\textit{P}3} belongs to the following set.
{\small
\begin{align*}
\textbf{S}=\lbrace(0,P_{max2}),(P_{max1},0),(P_{max1},P_{max2}),&(p_{k,d}^a,P_{max2}), \\
&(P_{max1},p_{j,u}^a)\rbrace
\end{align*}
{\normalsize where} 
\begin{align}
p_{k,d}^a=\frac{-B-\sqrt{B^2-4AC}}{2A} , p_{j,u}^a=\frac{-E-\sqrt{E^2-4DF}}{2D}
\end{align}
{\normalsize and}  
\begin{align}
A&=w_k g_k \beta^2 \\
B&=2w_kN_0g_k\beta+(w_k-v_j)\beta g_kg_jp_{j,u} \\
\nonumber C&=w_kg_kN_0^2+w_kg_kg_jN_0p_{j,u}-v_jN_kg_jp_{j,u}\beta \\
 &-v_jg_j\beta I_{k,j}p_{j,u}^2 \\
D&=v_j g_j I_{k,j}^2 \\
E&=2v_jN_kg_jI_{k,j}+(v_j-w_k)I_{k,j} g_kg_jp_{k,d} \\
\nonumber F&=v_jg_jN_k^2+v_jg_kg_jN_kp_{k,d}-w_kN_0g_kp_{k,d}I_{k,j} \\
&-w_kg_k\beta I_{k,j}p_{k,d}^2 
\end{align}
}
\end{prop}
\begin{proof}
Computing the derivative with respect to $p_{k,d}$ and setting it to zero we have:
\begin{equation*}
 \frac{\partial L}{\partial p_{k,d}}=0\Longrightarrow Ap_{k,d}^2+Bp_{k,d}+C=0
\end{equation*}
where $A$, $B$ and $C$ are defined above.
It is evident that $A\geq 0$, and if $w_k\geq v_j$ then $B\geq 0$. When $A,B\geq 0$ the  above quadratic equation  either  has  no zeros in $[0,P_{max1}]$ or has only one zero where the function changes sign from $-$ to $+$ indicating a local minimum for $L$. Therefore, in both cases the maximum is attained at a boundary point $0$ or $P_{max1}$. But when $w_k \leq v_j$, $B$ could be negative, and the smaller root of the quadratic equation $p_{k,d}^a$ could be positive. In this case, the maximum is attained  at $P_{max1}$ or $p_{k,d}^a$. By similar analysis for $p_{j,u}$ one sees that if $v_j\geq w_k$ then the maximum is attained at a boundary point $0$ or $P_{max2}$ and when $w_k\geq v_j$ the maximum is attained at $P_{max2}$ or $p_{j,u}^a$.
As a result, when $B\geq0$ the optimal transmission powers belong to the 
following set,
\small{
\begin{equation*}
P_{opt1}=\left\{(0,P_{max2}),(P_{max1},0),(P_{max1},P_{max2}),(P_{max1},p_{j,u}^a)\right\}.
\end{equation*}
}
\normalsize
Otherwise, if $B<0$, they belong to the set below
{\small 
\begin{equation*}
P_{opt2}=\left\{(0,P_{max2}),(P_{max1},0),(P_{max1},P_{max2}),(p_{k,d}^a,P_{max2})\right\}.
\end{equation*}
}
The cases {\small$(0,p_{j,u}^a)$} and {\small$(p_{k,d}^a,0)$} cannot be the optimal solutions of  \textbf{\textit{P}3} , because they are dominated by {\small$(0,P_{max2})$} and {\small$(P_{max1},0)$} which give a larger $L$.
\normalsize
Therefore, optimal powers could be found by checking the members of the set \textbf{S} and picking the one that corresponds to the largest $L$.
\end{proof}
Based on the above Proposition one can find the best uplink-downlink pair in each sub-channel by choosing the one with the largest value of $L$. This involves only $O(K^2)$ operations.
Now we can  present our sub-channel allocation algorithm to solve Problem \textbf{\textit{P}2}, in which we employ a sub-optimum power allocation scheme.
First, for each sub-channel $n$, we find the best channel gain among all users and denote it by $\tilde{g}({n})=\operatorname*{arg\,max}_k g_k(n) $. Then, we sort the sub-channels based on the value of $\tilde{g}({n})$. In other words. we find a sub-channel permutation $\{a_1,...,a_N\}$ such that
$\tilde{g}({a_1}) \geq \tilde{g}({a_2}) \geq . . .\geq \tilde{g}({a_N})$. Then, starting from sub-channel $a_1$, we seek $k$ and $j$ that maximize $L$. At the first iteration, we set $P_{max1}=P_0$ , $P_{max2}=P_k$ and for iteration $l\geq 2$ set $P_{max1}= \frac{P_0}{d_0(l)}$, $ \frac{P_k}{d_k(l)}$ where $d_0(l)$ and $d_k(l)$ indicate the number of sub-channels to be allocated to the BS's downlink transmission and to user $k$'s uplink transmission, respectively, in the $l$th iteration. The proposed sub-channel allocation algorithm is summarized below.

\begin{table}[h!]
 \begin{tabular}{m{30em}}
 \hline
 \textbf{Sub-channel Allocation Algorithm} \\
 \hline\hline
1.\textbf{for} $n=1$ to $N$ \textbf{do} \\ 
2.\quad $\tilde{g}({n})=\operatorname*{max}_k g_k(n) $
  \\
3.\textbf{end for} \\
4.Find a sub-channel permutation $\{ a_1,...,a_N \}$, $a_i\in\{1,...,N\}$, $a_i\neq a_j$ such that\\
\hspace{0.5cm} $\tilde{g}({a_1}) \geq \tilde{g}({a_2}) \geq . . .\geq \tilde{g}({a_N})$ \\
5. set $d_k(l)=1$  for $0\leq k \leq K$  and $1 \leq l \leq N$ \\
6.\textbf{for} $l=1$ to $N$ \textbf{do} \\
7.\quad Set $P_{max1}=\dfrac{P_0}{d_0(l)}$  and  $P_{max2}=\dfrac{P_k}{d_k(l)}      \forall k$  \\
8. \quad \textbf{for} $k=1$ to $K$ \textbf{do} \\
9. \quad \quad 	\textbf{for} $j=1$ to $K$ (if $k$ is an HD user $j\neq k$) \\
10.\quad \quad \quad In sub-channel $a(l)$ solve the problem  \textbf{\textit{P}3}  \\
11.\quad \quad \textbf{end for}
\\
12.\quad  \textbf{end for}
\\
13.\quad Using the obtained optimal powers, find the best pair $(k^*,j^*)$ in the \\
 \quad \quad sub-channel $a_l^*$ that  has the largest value of $L$
\\
14.\quad  $S_{j^*,u} \leftarrow [{S_{j^*,u},i}]$   ,   $S_{k^*,d} \leftarrow [{S_{k^*,d},i}]$ \\
15.\quad \textbf{if} $p_{k^*}\neq 0$ then $d_0(n)=d_0(n)+1$; \\
16.\quad \textbf{if} $p_{j^*}\neq 0$ then $d_{j^*}(n)=d_{j^*}(n)+1$; \\
17.\textbf{end for} \\
 \hline
 \end{tabular}
\end{table}
The complexity of finding the best user in each sub-channel is $O(K)$ and for $N$ sub-channels is $O(KN)$. Similarly, the complexity of finding the best pair in each sub-channel is $O(K^2)$ and doing so for $N$ sub-channels requires $O(NK^2)$ operations. Since the complexity of sorting $N$ values is $O(N \log N)$, then the overall computational complexity of the proposed sub-channel allocation algorithm is $O(N \log N+NK^2)$.
\section{Power Allocation}        \label{power_allocation}
The power allocation problem, denoted by \textbf{\textit{P}4}, can be formulated as follows 
\begin{align*}
\textbf{\textit{P}4}: & \operatorname*{maximize}_{p_{k,d},p_{j,u}} \qquad &R_d+R_u\\
& \text{subject to} \qquad &\text{(4)-(6)}
\end{align*}

Due to the interference terms, the power allocation problem is non-convex. Here, we use the ``difference of two concave 
functions/sets'' (DC) programming technique \cite{tuy2013convex} to convexify this problem. In this procedure, the non-concave 
objective function is expressed as the difference of two concave 
functions, and the discounted term is approximated by its first order 
Taylor series. Hence, the objective becomes concave and can 
be maximized by known convex optimization methods. 
This procedure runs iteratively, and after each iteration the 
optimal solution serves as an initial point for the next iteration 
until the improvement diminishes in iterations. In \cite{kha2012fast}, the DC approach is used to formulate optimized power allocation in a multiuser interference channel. Here, we rewrite the objective function of \textbf{\textit{P}4} in DC form as follows
\vspace{-2mm}
\begin{align*}
\operatorname*{max}_{\mathbf{p}} \quad 
f(\mathbf{p})-h(\mathbf{p}) 
\end{align*}

 \vspace{-9mm}
\small{
\begin{align*}
f(\mathbf{p})&=\sum_{k=1}^{K}\sum_{n\in S_{k,d}}w_k\log(N_k+I_{k,j}(n)p_{j,u}(n)+ g_k(n)p_{k,d}(n)) \\
&+\sum_{j=1}^{K}\sum_{n\in S_{j,u}}v_j\log(N_0+\beta p_{k,d}(n)+ g_j(n)p_{j,u}(n)) 
\end{align*}
\vspace{-7mm}
\begin{align*}
h(\mathbf{p})&=\sum_{k=1}^{K}\sum_{n\in S_{k,d}}w_k\log(N_k+I_{k,j}(n)p_{j,u}(n))\\
&+\sum_{j=1}^{K}\sum_{n\in S_{j,u}}v_j\log(N_0+\beta p_{k,d}(n)) 
\end{align*}
}
\normalsize
where  
\begin{align*}
\mathbf{p}=[p_{k_1,d}(1),...,p_{k_N,d}(N),p_{j_1,u}(1),,...,p_{j_N,u}(N)]^T 
\end{align*} is the downlink and uplink transmitted power vector, and $k_i$ and $j_i$ denote the uplink and downlink users that has been selected for the $i$th sub-channel after the sub-channel allocation phase. Now, the objective $f(\mathbf{p})-h(\mathbf{p})$ is a DC function. To write the Taylor series of the discounted function $h(\mathbf{p}$), we need its gradient, that can be easily derived as follows.
\small{
\begin{align*}
&\nabla h(\mathbf{p})= \bigg[\frac{u_{j_1}\beta}{\ln(2)}\frac{1}{N_0+\beta p_{k_1,d}(1))},. . . ,\frac{u_{j_N}\beta}{\ln(2)}\frac{N}{N_0+\beta p_{k_N,d}(N))}, \\ &\frac{w_{k_1}I_{k_1,j_1}(1)}{\ln(2)}\frac{1}{N_{k_1}+I_{k_1,j_1}(1)p_{j_1,u}(1))},..., \\
&\frac{w_{k_N}I_{k_N,j_N}(N)}{\ln(2)}\frac{1}{N_{k_N}+I_{k_N,j_N}(N)p_{j_N,u}(N))} \bigg] ^T 
\end{align*}
}

\normalsize
To make the problem convex, $h(\mathbf{p})$ is approximated with its first order approximation $h(\mathbf{p}^{(t)})+\nabla h^T(\mathbf{p}^{(t)})(\mathbf{p}-\mathbf{p}^{(t)})$ at point $\mathbf{p}^{(t)}$. We start from a feasible $\mathbf{p}^{(0)}$ at the first iteration, and $\mathbf{p}^{(t+1)}$ at the $t$th iteration is generated as the optimal solution of the following convex program
\begin{align*}
\mathbf{p}^{(t+1)}=\operatorname*{arg max}_{\mathbf{p}} \quad &f(\mathbf{p})- h(\mathbf{p}^{(t)})-\nabla h^T(\mathbf{p}^{(t)})(\mathbf{p}-\mathbf{p}^{(t)}) \\ 
&\text{subject to } \ (4)-(6)
\end{align*}
Since $h(\mathbf{p})$ is a concave function, its gradient is also its super gradient so we have
\begin{align*}
h(\mathbf{p})\leq h(\mathbf{p}^{(t)})+\nabla h^T(\mathbf{p}^{(t)})(\mathbf{p}-\mathbf{p}^{(t)}) 
\end{align*}
and we can deduce
\begin{align*}
h(\mathbf{p}^{(t+1)})\leq h(\mathbf{p}^{(t)})+\nabla h^T(\mathbf{p}^{(t)})(\mathbf{p}^{(t+1)}-\mathbf{p}^{(t)}). 
\end{align*}
Then it can be proved that in each iteration the solution of  problem \textbf{\textit{P}4} is improved as follows
\begin{align*}
&f(\mathbf{p}^{(t+1)})- h(\mathbf{p}^{(t+1)}) \geq \\ &f(\mathbf{p}^{(t+1)})- h(\mathbf{p}^{(t)})-\nabla h^T(\mathbf{p}^{(t)})(\mathbf{p}^{(t+1)}-\mathbf{p}^{(t)}) \\
&=\operatorname*{max}_{\mathbf{p}} \quad f(\mathbf{p})- h(\mathbf{p}^{(t)})-\nabla h^T(\mathbf{p}^{(t)})(\mathbf{p}-\mathbf{p}^{(t)}) \\
& \geq f(\mathbf{p}^{(t)})- h(\mathbf{p}^{(t)})-\nabla h^T(\mathbf{p}^{(t)})(\mathbf{p}^{(t)}-\mathbf{p}^{(t)}) \\
&=f(\mathbf{p}^{(t)})- h(\mathbf{p}^{(t)}).
\end{align*}
According to the above equations, the objective value after each iteration is either unchanged or improved and since the constraint set is compact it can be concluded that the above DC approach converges to a local maximum.

\section{Simulation Results}
\label{simulation_results}
In this Section, we evaluate the proposed resource allocation scheme for OFDMA networks with half-duplex and imperfect full-duplex nodes in indoor and outdoor scenarios.  We assume a time-slotted system, where nodes are uniformly distributed within a given cell radius. Table II presents the details of the indoor and outdoor simulation setup and channel models. In addition to the pathloss, a Rayleigh block fading channel model with unit average power is considered. The channel gains remain constant in each time slot and vary independently from one time slot to the next.  
For comparison, we consider six schemes: (i) An HD uplink system (HD-U), (ii) An HD downlink system (HD-D), (iii) a system that includes an FD BS and  HD users (FD-HD), (iv) a system that contains  an FD BS and FD users (FD-FD), (v) an upper bound which is HD uplink rate plus HD downlink rate; (vi) a Hybrid HD scheme (HHD), in which a hybrid HD BS could transmit data to downlink users and receive data from uplink users simultaneously in different sub-channels. 
For the HD-D case, each sub-channel is allocated to the user with the best weighted channel SNR, and multi-level water filling \cite{seong2006optimal} is applied  for power allocation. For the sub-channel assignment of the HD-U scheme the $SOA1 \: 4B \: 5A$ method presented in \cite{huang2009joint} is used, and for power allocation each user performs water filling in its dedicated sub-channels. In the HHD scheme, we use the proposed sub-channel allocation algorithm by changing the set $P_{opt}$ to:
 \begin{align*}
P_{opt}=\left\{(0,P_{max2}),(P_{max1},0)\right\}
\end{align*}
and perform multi-level water filling and water filling for the power allocation in the selected downlink and uplink sub-channels, respectively.

Fig. \ref{fig:convergence} illustrates the convergence of the proposed resource allocation scheme in an OFDMA network with 10 HD nodes and 10 imperfect FD nodes. As can be seen, the sum-rate converges in just a few iterations.

Fig. \ref{fig:Compare_Optimal} compares the proposed algorithm with the optimal exhaustive search solution. Due to the high computational complexity of exhaustive search, only a small network with one HD and one FD user and a small number of subchannels can be considered. Uplink and downlink weight vectors are assumed to be $\mathbf{u}=[1/3, 2/3]^{T}$ and $\mathbf{w}=[2/3, 1/3]^{T}$ respectively, and the SI cancellation coefficient is set to $\beta = -90$ dB. Simulation results show that, at least for small size networks, our proposed algorithm achieves the performance of the optimal exhaustive search.
\begin{figure}
  \includegraphics[width=\linewidth, height=4cm]{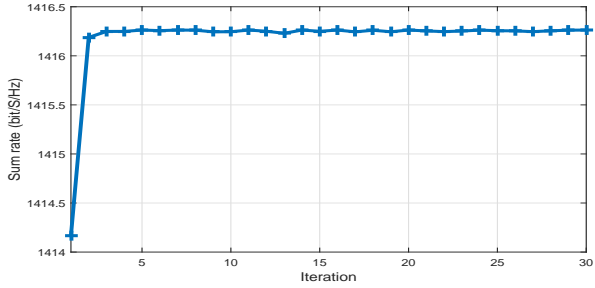}
  \caption{Convergence behavior of the proposed algorithm. Here $\beta=10^{-6}$ and other simulation parameters are the same as in the outdoor case}
    \label{fig:convergence}
\end{figure}

\begin{figure}[t!]
  \includegraphics[width=\linewidth, height=4cm]{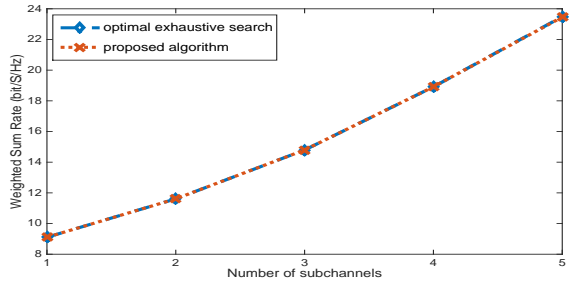}
  \caption{Performance comparison of the proposed algorithm with optimal exhaustive search in a small network}
    \label{fig:Compare_Optimal}
\end{figure}

\begin{figure}[t!]
  \includegraphics[width=\linewidth, height=4cm]{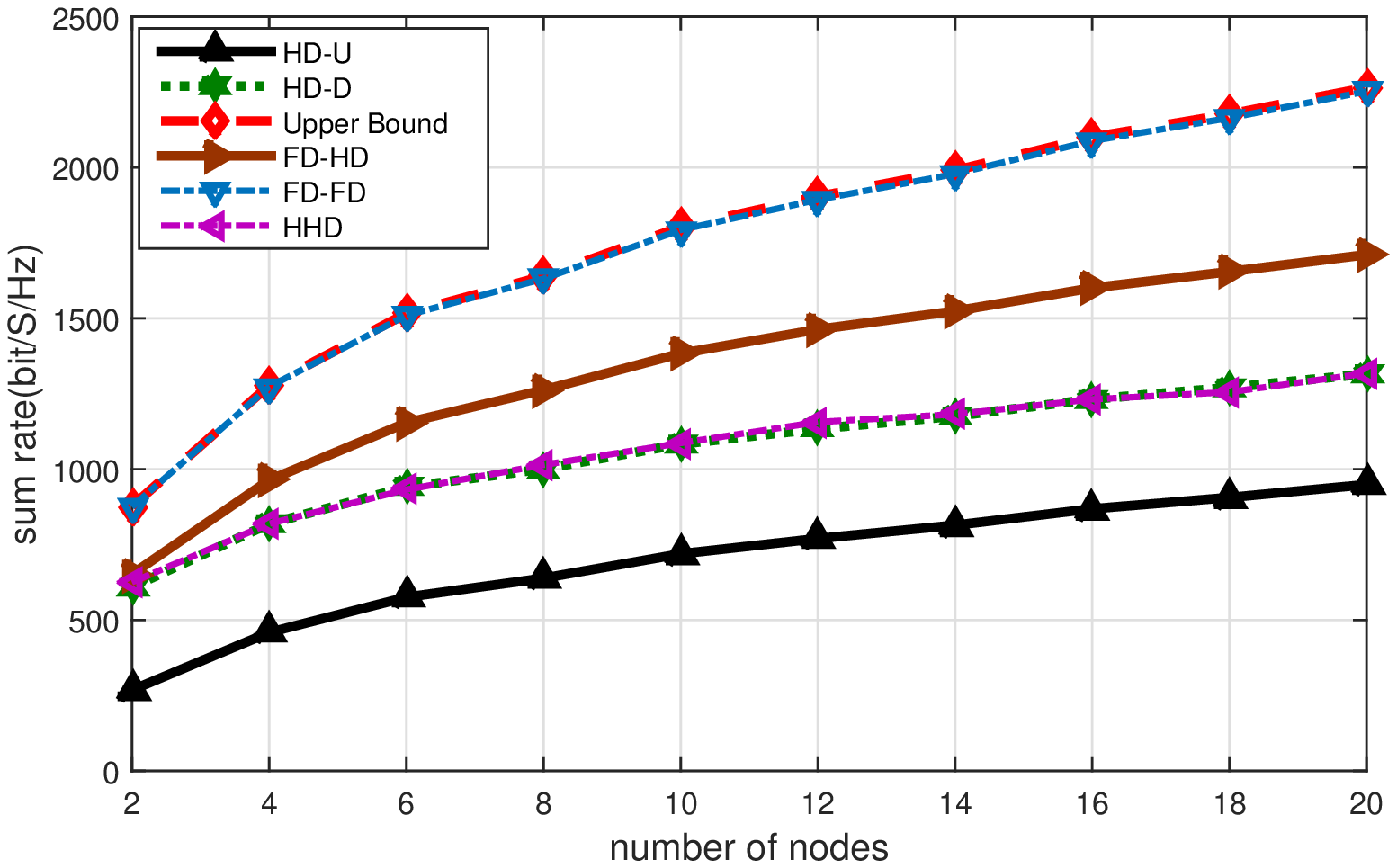}
  \caption{Performance comparison of six schemes in FD and HD networks in the outdoor scenario.}
    \label{fig:all_outdoor}
\end{figure}

\begin{figure}[t!]
  \includegraphics[width=\linewidth, height=4cm]{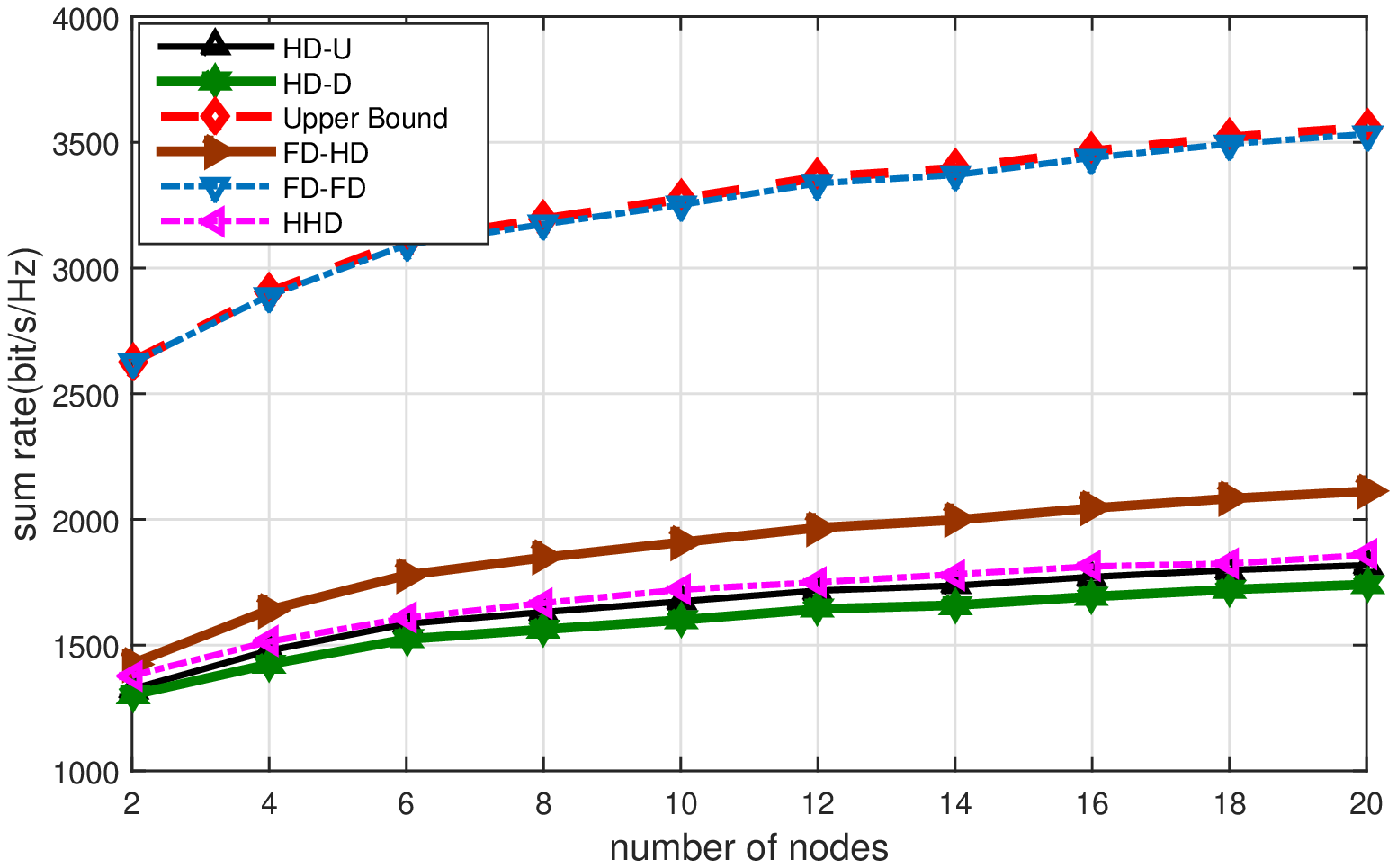}
  \caption{Performance comparison of six schemes in FD and HD networks in the indoor scenario.}
      \label{fig:all_indoor}
\end{figure}

\begin{table}[h!]
 \caption{simulation parameters}
 \begin{tabular}{|m{14em} | m{15em}|}

 \hline
 \textbf{PARAMETER} & \textbf{VALUE} \\
 \hline
 Maximum BS Power (outdoor) & $43$ dBm \\
 \hline
 Maximum BS Power (indoor) & $24$ dBm \\
 \hline
  Maximum UE Power $(P_k)$ & $23$ dBm \\
 \hline
  Thermal Noise Density & $-170$ dBm/Hz\\
 \hline
  Number of Sub-channels $N$ & $64$\\
 \hline
  Total Bandwidth & $10$ MHz \\
 \hline
  Sub-channel Bandwidth  & $150$ KHz \\
  \hline
  Cell Radius (outdoor) & $1$ km\\
   \hline
  Cell Radius (indoor) & $20$ m\\
 \hline
 Center Frequency & $2$ GHz\\
 \hline
 BS to UE Pathloss (outdoor) & urban Hata model with parameters                                  $h_m= 1.5$ m, $h_B= 30$ m \\   
 \hline
 UE to UE Pathloss (outdoor) & urban Hata model with parameters                                  $ h_m= 1.5$ m, $h_B= 1.5$ m \\
 \hline
  Pathloss Model (indoor) & ITU model for indoor attenuation with parameters N= $22$, $p_f(n)= 9$ \\   
 \hline
 \end{tabular}
\end{table}

Fig. \ref{fig:all_outdoor} shows the sum-rate of the different schemes in the outdoor scenario with perfect SI cancellation ($\beta=0$). It can be seen that when the BS and all nodes are perfect FD devices the upper-bound could be attained, and when the nodes are HD but the BS is FD the sum-rate is still bigger than the cases with HD BS, but it can not reach the upper-bound because of inter-node interference.

Fig. \ref{fig:all_indoor} shows the sum-rate of the six presented schemes in an indoor scenario. If we compare the outdoor 
and indoor scenarios we find that using an FD BS in an outdoor environment has much larger gain  than  using it in an indoor case. This result is  intuitive because in the outdoor environment the distances between  nodes are larger, and hence the inter-node interference is smaller. As a result, the FD BS could work in FD mode in more sub-channels, which helps increase the network throughput more significantly.

\begin{figure}[t!]
  \includegraphics[width=\linewidth, height=4cm]{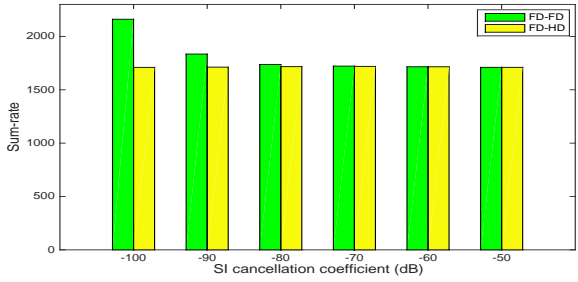}
  \caption{Effect of the self interference cancellation coefficient on the FD network capacity in the indoor scenario.}
        \label{fig:beta_threshold_indoor}
\end{figure}

 Fig. \ref{fig:beta_threshold_indoor} compares the sum-rates of an FD-FD network and an FD-HD network for different values of $\beta$ in the indoor scenario. It can be seen that when $\beta$ is larger than a specified threshold, which is near $-90$ dB, there is no difference between the sum-rate of the all HD user case and the sum-rate of the all FD user case. The reason is that when $\beta$ is large relative to the inter-node interference, FD users prefer to work in HD mode in order to increase their rate, hence the sum-rates of FD-FD and FD-HD become equal.  

In Fig. \ref{fig:beta_threshold_outdoor}, the same experiment is repeated for the outdoor scenario. Here the threshold $\beta$ is approximately $-110$ dB which is much smaller than in the indoor case. Since the inter-node interference in the outdoor environment is smaller, the SI cancellation coefficient should be very small to make the FD users work in the FD mode. Therefore, in the outdoor scenario a full-duplex user should be almost a perfect FD in order to be allowed to work in the FD mode.

Fig. \ref{fig:FDpercentage} shows the performance of a full-duplex OFDMA network with a mix of FD and HD users. A total of 20 users are considered, assuming perfect SI cancellation  for FD devices. It can be seen that increasing the percentage of FD users in an outdoor environment does not increase the total sum-rate significantly, but in the indoor case by equipping only $10\%$ of the nodes with FD technology the network throughput greatly increases. The reason behind this is the large inter-node interference in the indoor environment that could be avoided by using FD users instead of HD ones. 
\begin{figure}[t!]
  \includegraphics[width=\linewidth, height=4cm]{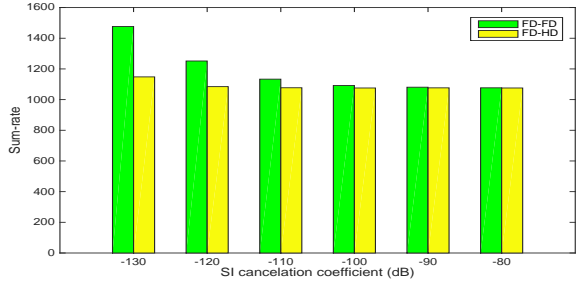}
  \caption{Effect of the self interference cancellation coefficient on the FD network capacity in the outdoor scenario}
          \label{fig:beta_threshold_outdoor}
\end{figure}
\begin{figure}[t!]
  \includegraphics[width=\linewidth, height=4cm]{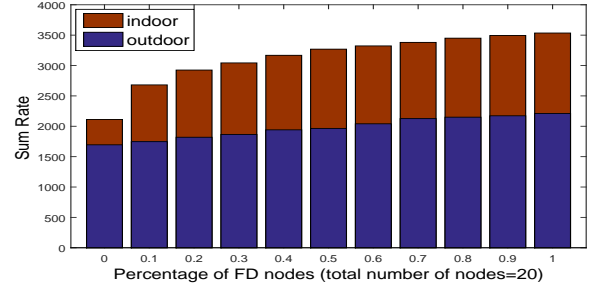}
  \caption{Effect of the fraction of FD nodes on the network throughput in both indoor and outdoor scenarios}
  \label{fig:FDpercentage}
\end{figure}

\section{conclusion}
\label{conclusion}
In order to fully exploit the advantages of FD technology in wireless networks, it is important to design an appropriate resource allocation algorithm. In this paper we considered a single cell OFDMA network that contains an FD BS and a mixture of HD and FD users, and also assumed that FD nodes are not  necessarily perfect FD devices and may suffer from residual self-interference. For this model, we proposed a sub-channel allocation algorithm and power allocation method and showed that when all users and the BS are perfect FD we can double the capacity. Otherwise, because of inter-node interference and self-interference the spectral efficiency gain is smaller, but we showed that even by using an imperfect FD BS in a network, the total capacity could  increase significantly. We also investigated FD operation in both outdoor and indoor scenarios and studied the effect of the self interference cancellation coefficient and of the percentage of FD users.

\bibliographystyle{IEEEtran}
{\small
\bibliography{bibfiles}}
\end{document}